\documentclass[conference]{IEEEtran}
\ifCLASSINFOpdf
\else
\fi
\usepackage{setspace}
\usepackage{verbatim}
\usepackage{cite}
\usepackage{amsmath,amssymb,amsfonts,amsthm}
\usepackage{algorithm}
\usepackage{algpseudocode}
\usepackage{graphicx}
\usepackage{epstopdf}
\usepackage{comment}
\usepackage{textcomp}
\usepackage{tikz}
\usepackage{color}
\usepackage{stfloats}

\newtheorem{lemma}{{Lemma}}

\def\BibTeX{{\rm B\kern-.05em{\sc i\kern-.025em b}\kern-.08em
    T\kern-.1667em\lower.7ex\hbox{E}\kern-.125emX}}
\usepackage{bm}      
\usepackage{multirow}
\usepackage{algpseudocode}
\usepackage{amsmath,amssymb}
\usepackage{tikz}
\usetikzlibrary{calc,arrows.meta}
\usetikzlibrary{positioning,calc}
\usepackage{epstopdf}
\usepackage[top=0.75in, bottom=1.09in, left=0.7in, right=0.7in]{geometry}
\allowdisplaybreaks[4]

\begin{document}

\title{TOIB: Task-Oriented Orthogonalised Information Bottleneck for Distributed Semantic Communication}
\author{
\IEEEauthorblockN{Jiaxiang Wang\IEEEauthorrefmark{1},
Zhaohui Yang\IEEEauthorrefmark{2}\IEEEauthorrefmark{3},
Yahao Ding,\IEEEauthorrefmark{1}, Ye Hu\IEEEauthorrefmark{4}, 
and Mohammad Shikh-Bahaei\IEEEauthorrefmark{1}}
	\IEEEauthorblockA{
 		$\IEEEauthorrefmark{1}$Department of Engineering 
 King's College London 
London, UK \\
			$\IEEEauthorrefmark{2}$College of Information Science and Electronic Engineering, Zhejiang University, Hangzhou, China\\
   	$\IEEEauthorrefmark{3}$Zhejiang Provincial Key Laboratory of Info. Proc., Commun. \& Netw. (IPCAN), Hangzhou, China\\
      $\IEEEauthorrefmark{4}$ Department of Industrial and Systems Engineering, University of Miami, Miami, FL, USA\\
            E-mails: 
\{jiaxiang.wang, k19005540, m.sbahaei\}@kcl.ac.uk,
yang\_zhaohui@zju.edu.cn,
yehu@miami.edu
		}}

\maketitle

\begin{abstract}
Task-oriented semantic communication emerges as a crucial paradigm for next-generation wireless networks, aiming to efficiently transmit task-relevant information while reducing interference and redundancy across multiple users. Existing information bottleneck (IB)-based frameworks predominantly focus on single-user scenarios, neglecting cross-user semantic interference in distributed semantic communications. To overcome this limitation, we propose a task-oriented orthogonalised information bottleneck (TOIB) approach, explicitly designed for distributed semantic communication systems. By introducing task-conditioned latent variables, TOIB adaptively balances semantic sufficiency, semantic compression, and inter-user semantic orthogonality. Extensive simulations conducted on classification tasks demonstrate that TOIB consistently achieves superior classification accuracy across various signal-to-noise ratio (SNR) regimes compared to traditional IB and deep joint source-channel coding (JSCC) methods. Specifically, the proposed method significantly enhances robustness under harsh low-SNR conditions and effectively suppresses cross-user semantic interference, as validated by cross-decoding accuracy metrics.
\end{abstract}

\section{Introduction}
The proliferation of intelligent applications and multimedia services in future wireless networks necessitates a paradigm shift from traditional data-oriented communication towards semantic and task-oriented communication \cite{saad2019vision,zhao2025compression,yang2025integrated,ding2023distributed}. Unlike conventional communication systems focused solely on transmitting raw data bits with high reliability, task-oriented communication aims to deliver semantic features that are explicitly relevant to specific user tasks, such as image reconstruction, video analytics, and semantic classification. 

Existing research on deep learning (DL)-based semantic communication system has already demonstrated a favourable performance in single-user or point-to-point tasks, spanning a range of applications such as text-based \cite{xie2021deep}, image-based \cite{wang2025semantic}, and video-based \cite{wang2025generative} semantic communication system. However, the joint optimisation between the communication overhead and task execution performance was not considered by above methods. To address this issue, the authors in \cite{shao2021learning} have leveraged the information bottleneck (IB) framework to explicitly formulate the joint optimisation of communication overhead and task execution performance as a rate-distortion optimisation problem. The IB framework was later extended to multi-device cooperative edge inference in \cite{shao2022task}
However, traditional IB-based semantic communication framework only balances semantic compression and information sufficiency in single-user settings, but lack the capability to handle cross-user semantic interference in distributed semantic communication. 
However, directly adopting traditional IB-based methods which balance semantic compression and information sufficiency for single-user settings, the issue of semantic interference is not be well tackled when extended to distributed semantic communication.

Some recent works \cite{ma2025semantic,ma2023features} have focused on the IB based multi-user broadcast communication. In \cite{ma2025semantic}, the authors proposed a semantic feature division multiple access scheme for digital semantic broadcast scheme. However, this approach does not explicitly impose orthogonality constraints across different users’ semantic features. As a result, the burden of learning orthogonality is shifted entirely onto the network during training, thereby increasing the overall training complexity. In \cite{ma2023features}, the authors proposed a feature-disentangled semantic broadcast communication systems. However, this method does not account for the task correlation and heterogeneity.
In broadcast scenarios, if one directly minimises the compression term of each user and maximises its task-related mutual information in parallel without explicitly constraining the cross-user representational correlation, substantial semantic interference will arise in the latent space. Particularly, (i) for mutually independent tasks, the semantic features of different users become entangled in the latent space, causing irrelevant or even conflicting semantic components to be inadvertently carried over and propagated to other users’ receivers; (ii) for correlated tasks, beneficial task structures exist across users that can be jointly exploited, and excessive separation of these features would hinder such cross-task knowledge transfer.


Motivated by these issues, this paper proposes a task-oriented semantic orthogonalised information bottleneck (TOIB) framework to jointly balance semantic sufficiency, compression, and inter-user orthogonality. The main contributions of this work are summarised as follows:
\begin{itemize}
    \item We establish an end-to-end single-cell distributed semantic communication system, where the BS employs user-specific semantic encoders to extract latent semantic representations, performs power-controlled superposition, and transmits over a wireless channel; each user decodes its task target from the received signal via a dedicated semantic decoder.
    \item Under this broadcast setting, we formulate a task-oriented orthogonalised IB optimisation to simultaneously maximise semantic sufficiency and enforce semantic compression, and we identify that naive parallel optimisation across users leads to strong latent-space coupling, inducing redundancy and task degradation due to semantic interference.
    \item To explicitly regulate cross-user dependence, we introduce a pairwise task-conditioning variable and propose the TOIB objective by augmenting IB with a conditional mutual-information regulariser. We further derive a tractable variational surrogate and employ CLUB-based estimators together with Monte Carlo sampling, resulting in a two-stage training procedure that alternates between tightening inter-user MI upper bounds and updating the main networks
    \item Extensive CIFAR-10 classification experiments demonstrate that TOIB yields consistently improved accuracy over the considered baselines.
    Deep VIB and Deep JSCC across SNR regimes. In particular, TOIB achieves clear gains on both AWGN and Rayleigh channels, including $2\%–5\%$ improvement over Deep VIB under Rayleigh fading, while Deep JSCC saturates at about $84\%$.
\end{itemize}

\vspace{-2mm}
\section{System Model and Problem Formulation}

\subsection{System Overview}

\begin{figure}[t]
 \centerline{\includegraphics[width=0.48\textwidth]{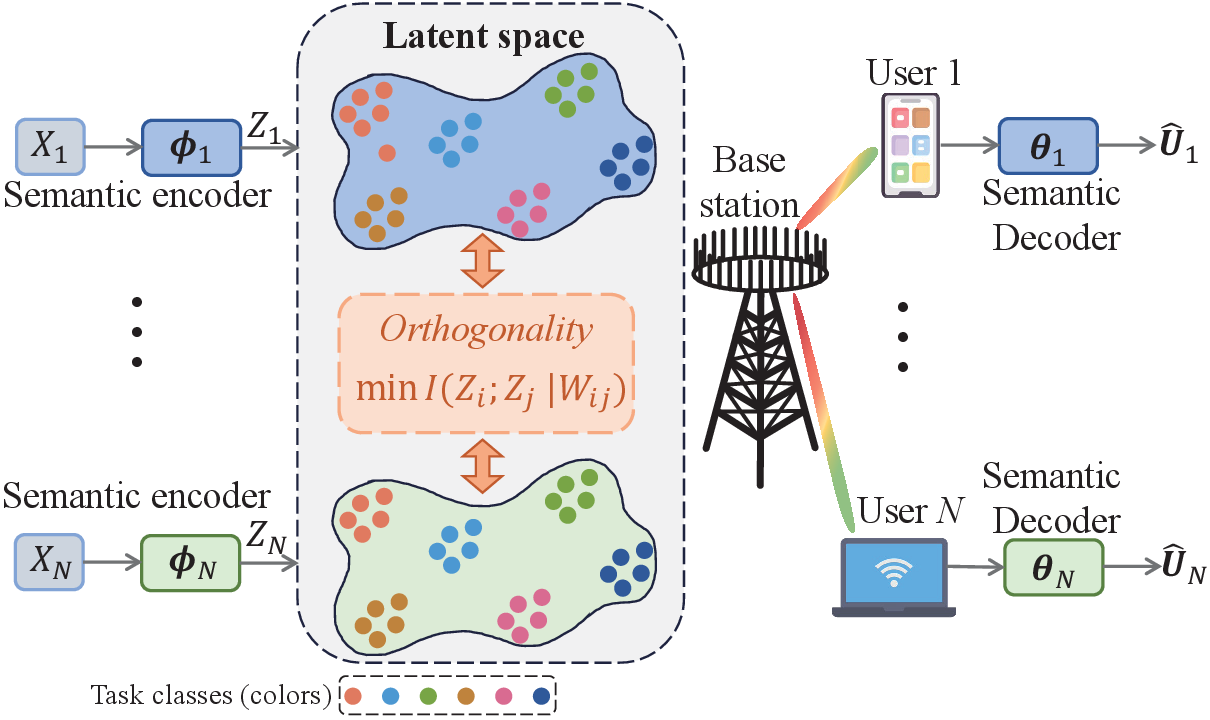}}
    \caption{The considered distributed semantic communication system.}
    \label{fig: system model}
\end{figure}

As shown in Fig. \ref{fig: system model}, we focuses on a single-cell task-oriented distributed semantic communication system. In this system, the base station (BS) concurrently transmits semantic information to the user set $\mathcal{N}=\{1,\ldots,N\}$, and all users share the identical time and frequency resources.
Specifically, the source information of each user $i\in\mathcal{N}$ is denoted as a random variable $X_i$, and the corresponding semantic task target is denoted as a random variable $U_i$. The BS is equipped with a user-specific semantic encoder parametrised by $\boldsymbol{\phi}_i$ to extract the semantic feature representation for each user image  $\boldsymbol{x}_i$, which is denoted as:
\begin{align}
\label{eq: encoder}
    \boldsymbol{z}_i \sim p_{\boldsymbol{\phi}_i}(\boldsymbol{z}_i \mid \boldsymbol{x}_i),
\end{align}
where $p_{\boldsymbol{\phi}_i}(\boldsymbol{z}_i \mid \boldsymbol{x}_i)=\mathcal{N}(\boldsymbol{\mu}_{\phi_i}(\boldsymbol{x}_i), \mathrm{diag}(\boldsymbol{\sigma}^2_{\phi_i}(\boldsymbol{x}_i)))$, and $\mathbb{E}[\|\boldsymbol{z}_i\|^2] = 1$.
Then, the BS transmitter performs power control and superimposes the semantic representations of all users to form a transmission signal, which is denoted as:
\begin{align}
\label{eq: superimposed signal}
    \boldsymbol{s} = \sum^{N}_{i=1} \sqrt{p_i}{\boldsymbol{z}}_i,
\end{align}
where $p_i$ is the power allocated to user $i$'s semantic information, with $\sum_{i=1}^N p_i = P^{\max}$, and $ P^{\max}$ being the maximum power of BS.
The received signal of the downlink user $i\in \mathcal{N}$ is denoted as:
\begin{align}
\label{eq: received signal}
    \boldsymbol{y}_i = h_i \boldsymbol{s} + \boldsymbol{n}_i,
\end{align}
where $h_i$ is the channel gain from the BS to user $i$, and $n_i \sim \mathcal{N}(0, \sigma^2_i\mathbf{I})$ is the Additive White Gaussian Noise (AWGN) of user $i$.
After receiving the signal, the receiver at user $i$ decodes the corresponding user task target from $\boldsymbol{y}_i$ using its dedicated semantic decoder $\boldsymbol{D}_{\boldsymbol{\theta}_i}(\cdot)$ by:
\begin{align}
\label{eq: decoder}
    \hat{\boldsymbol{u}}_i = \boldsymbol{D}_{\boldsymbol{\theta}_i}({\boldsymbol{y}}_i),
\end{align}
where $\boldsymbol{D}_{\boldsymbol{\theta}_i}(\cdot)$ denotes the semantic decoder of user $i$ with learnable parameters $\boldsymbol{\theta}_i$.

According to the end-to-end broadcast model, the transmit signal is a deterministic function of all users' latent representations, and each receiver observes the signal \eqref{eq: received signal}. We optimise per-user task performance under a semantic compression constraint, while explicitly controlling cross-user dependence in the latent space through a conditional mutual information regulariser.
\vspace{-4mm}
\subsection{Problem Formulation}
Our goal is to find the optimal semantic encoder and decoder design that can serve all the users with balanced semantic sufficiency and compression in the considered task-oriented communication systems. Thus, we formulate an optimization problem whose objective is to maximize the semantic sufficiency and compression at all users:
\begin{align}
    \underset{\{\boldsymbol{\phi}_i,\boldsymbol{\theta}_i\}^N_{i=1}}{\min}  \sum^{N}_{i=1}[\underbrace{-I(U_i; Y_i)}_{\text{\textbf{sufficiency}}} + \underbrace{\beta I(X_i; Z_i)}_{\text{\textbf{compression}}} ]. \label{eq: conventional IB}
\end{align}
However, under the conventional IB framework, if one merely minimises the single-user semantic compression term and maximises the semantic sufficiency term without imposing constraints on cross-user representation correlation, the semantic representations from different users in a downlink distributed communication scenario often exhibit strong coupling in the latent space. Such semantic interference leads to two consequences: on the one hand, it introduces information redundancy across tasks, causing the limited downlink bandwidth to be occupied by duplicated features; on the other hand, when the tasks are independent, it delivers irrelevant noisy features to each user, thereby significantly degrading classification performance.

\section{A TOIB Approach for Task-oriented Distributed Communication} \label{sec: TOIB method}
In this section, we introduce a novel TOIB approach to overcome the limitations of the conventional IB framework \eqref{eq: conventional IB} in multi-user distributed scenarios. 
We establish the proposed system framework and formulates its information-theoretic objective tailored for distributed semantic communication. 

\subsection{Proposed TOIB Scheme}
We introduce a pairwise conditioning variable $W_{ij}$ to capture task-related dependence between user $i$ and user $j$. 
In the classification experiments, we let $W_{ij}$ take discrete values in $\{1,\ldots,K\}$, corresponding to the class index, and we construct matched/mismatched cross-user pairs within each class to estimate $I(Z_i;Z_j\mid W_{ij})$.
The variable subsequently induces the conditional joint distribution of latent semantic representations as:
\begin{align}
    p(\boldsymbol{z}_i,& \boldsymbol{z}_j | \boldsymbol{w}_{ij}) =
    \notag \\ 
    &\iint p(\boldsymbol{x}_i, \boldsymbol{x}_j | \boldsymbol{w}_{ij}) p_{\boldsymbol{\phi}_i}(\boldsymbol{z}_i | \boldsymbol{x}_i) p_{\boldsymbol{\phi}_j}(\boldsymbol{z}_j | \boldsymbol{x}_j) d\boldsymbol{x}_id\boldsymbol{x}_j.
\end{align}
The information theoretic objective of proposed TOIB is expressed as:
\begin{align}
\label{eq: TOIB}
  &\mathcal{L}_{\text{TOIB}} \notag \\
  &= \sum^{N}_{i=1}[\underbrace{-I(U_i; Y_i)}_{\text{\textbf{sufficiency}}} + \underbrace{\beta I(X_i; Z_i)}_{\text{ \textbf{compression}}} +\alpha \sum^N_{\substack{j=1 \\ j \ne i}} \underbrace{I(Z_i; Z_j\mid W_{ij})}_{\text{\textbf{orthogonality}}} ], 
\end{align}
where the efficiency term $I(U_i; Y_i)$ encourages the received signal $Y_i$ to preserve as much task-relevant information $U_i$ as possible, the compression term $I(X_i;Z_i)$ penalizes the mutual information between the raw source $X_i$ and the semantic representation $Z_i$ to compress the task-irrelevant details of the source while retaining only the essential task-relevant semantics, and the orthogonality term $I(Z_i; Z_j\mid W_{ij})$ enforces orthogonality across the latent semantic representations of different users to minimise the conditional mutual information between the latent semantic representations of user $i$ and user $j$ given the shared task $W_{ij}$. The parameters $\alpha, \beta \geq 0$ are the weighting factors to control the degree of semantic compression and semantic orthogonality, respectively. The optimisation problem \eqref{eq: TOIB} can be equivalently reformulated as follows:
\begin{align}
    &\mathcal{L}_{\text{TOIB}} = \sum^N_{i=1} \left[ \mathbb{E}_{p_{\boldsymbol{\phi}_i}(\boldsymbol{u}_i, \boldsymbol{y}_i)}\left[-\log p_{\boldsymbol{\phi}_i}(\boldsymbol{u}_i \mid \boldsymbol{y}_i)\right] -H(U_i) \notag \right.  \\
    & \left. + \beta \mathbb{E}_{p(\boldsymbol{x}_i)}\left[D_{\text{KL}}(p_{\boldsymbol{\phi}_i}(\boldsymbol{z}_i \mid \boldsymbol{x}_i) \parallel p_{\boldsymbol{\phi}_i}(\boldsymbol{z}_i))\right] \right] \notag \\
    &  + \alpha \sum^N_{i=1} \sum^N_{\substack{j=1 \\ j \ne i}} \mathbb{E}_{p(\boldsymbol{w}_{ij})} \mathbb{E}_{p(\boldsymbol{z}_i, \boldsymbol{z}_j | \boldsymbol{w}_{ij})}\left[ \log \frac{ p(\boldsymbol{z}_j | \boldsymbol{z}_i, \boldsymbol{w}_{ij})}{p(\boldsymbol{z}_j | \boldsymbol{w}_{ij})} \right], \label{eq: loss TOIB}
\end{align}
where $D_{\text{KL}}(p_{\boldsymbol{\phi}_i}(\boldsymbol{z}_i \mid \boldsymbol{x}_i) \parallel p_{\boldsymbol{\phi}_i}(\boldsymbol{z}_i))$ is the Kullback–Leibler (KL) divergence between $p_{\boldsymbol{\phi}_i}(\boldsymbol{z}_i \mid \boldsymbol{x}_i)$ and $p_{\boldsymbol{\phi}_i}(\boldsymbol{z}_i)$, and the term $H(U_i)$ in \eqref{eq: loss TOIB} is constant for a given task and can be ignored in the optimisation process. 
However, optimising \eqref{eq: loss TOIB} faces the challenges of high-dimensional integrals in the latent prior $p_{\boldsymbol{\phi}_i}(\boldsymbol{z}_i) = \int p(\boldsymbol{x}_i)p_{\boldsymbol{\phi}_i}(\boldsymbol{z}_i | \boldsymbol{x}_i)$, the label posterior distribution $p_{\boldsymbol{\phi}_i}(\boldsymbol{u}_i | \boldsymbol{y}_i)$ and $p(\boldsymbol{z}_j | \boldsymbol{z}_i)$, which are computationally prohibitive. To address this challenge, we use the variational approximation method to derive a variational surrogate upper bound $\mathcal{L}_{\text{vTOIB}}$ in \textbf{Lemma \ref{theorem 1}}.

\begin{lemma} \label{theorem 1}
    The variational surrogate upper bound of \eqref{eq: loss TOIB} is given by:
    \begin{align}
        \mathcal{L}_{\text{TOIB}} &\leq \sum^{N}_{i=1} \left[ \mathbb{E}_{p_{\boldsymbol{\phi}_i}(\boldsymbol{u}_i, \boldsymbol{y}_i)}[-\log q_{\boldsymbol{\theta}_i}(\boldsymbol{u}_i \mid \boldsymbol{y}_i)] \right. \notag \\
        &+ \left. \beta \mathbb{E}_{p_{\boldsymbol{\phi}_i}(\boldsymbol{x}_i)}\left[D_{\text{KL}}(p_{\boldsymbol{\phi}_i}(\boldsymbol{z}_i \mid \boldsymbol{x}_i) \parallel r_i(\boldsymbol{z}_i))\right] \right] \notag \\
        &+ \alpha \sum^N_{i=1} \sum^N_{\substack{j=1 \\ j \ne i}} \mathbb{E}_{p(\boldsymbol{w}_{ij})} \left[\mathbb{E}_{p(\boldsymbol{z}_i,\boldsymbol{z}_j|\boldsymbol{w}_{ij})}\log p(\boldsymbol{z}_j|\boldsymbol{z}_i, \boldsymbol{w}_{ij}) \notag \right. \\
        & \left. - \mathbb{E}_{p(\boldsymbol{z}_i|\boldsymbol{w}_{ij})p(\boldsymbol{z}_j|\boldsymbol{w}_{ij})}\log p(\boldsymbol{z}_j|\boldsymbol{z}_i, \boldsymbol{w}_{ij})\right], \label{eq: vTOIB theorem}
    \end{align}
    and \eqref{eq: loss TOIB} can be accurately approximated by $\mathcal{L}_{\text{vTOIB}}$ in \eqref{eq: surrogate} by replacing $p(\boldsymbol{z}_i,\boldsymbol{z}_j | \boldsymbol{w}_{ij})$ in \eqref{eq: vTOIB theorem} with its variational approximation $q_{\boldsymbol{\psi}_{ij}}(\boldsymbol{z}_j|\boldsymbol{z}_i, \boldsymbol{w}_{ij})$,
    where $q_{\boldsymbol{\theta}_i}(\boldsymbol{u}_i | \boldsymbol{y}_i)$, $r_i(\boldsymbol{z}_i)$, $q_{\boldsymbol{\psi}_{ij}}(\boldsymbol{z}_j|\boldsymbol{z}_i, \boldsymbol{w}_{ij})$ are the variational approximation of the distribution $p_{\boldsymbol{\phi}_i}(\boldsymbol{u}_i | \boldsymbol{y}_i)$, $p_{\boldsymbol{\phi}_i}(\boldsymbol{z}_i)$, $p(\boldsymbol{z}_j|\boldsymbol{z}_i, \boldsymbol{w}_{ij})$, respectively, i.e. $p(\boldsymbol{u}_i | \boldsymbol{y}_i) \approx q_{\boldsymbol{\theta}_i}(\boldsymbol{u}_i | \boldsymbol{y}_i)$, $r_i(\boldsymbol{z}_i) \approx p_{\boldsymbol{\phi}_i}(\boldsymbol{z}_i)$, $q_{\boldsymbol{\psi}_{ij}}(\boldsymbol{z}_j | \boldsymbol{z}_i, \boldsymbol{w}_{ij}) \approx p(\boldsymbol{z}_j | \boldsymbol{z}_i, \boldsymbol{w}_{ij})$. Here, $\boldsymbol{\phi}_i$ and $\boldsymbol{\theta}_i$ are the parameters of the user $i$'s semantic encoder and decoder, respectively, and $\boldsymbol{\psi}_{ij}$ is the parameter of a Contrastive Log-ratio Upper Bound (CLUB)-based lightweight neural network to variationally approximate the distribution $p(\boldsymbol{z}_j | \boldsymbol{z}_i, \boldsymbol{w}_{ij})$.
\end{lemma}

\proof Please refer to the Appendix \ref{appendix A}.

\textbf{Lemma \ref{theorem 1}} derives a tractable variational upper bound of \eqref{eq: TOIB} by replacing its three intractable distributions with learnable surrogates, enabling end-to-end gradient-based optimisation.

\begin{table*}[!t]
\centering
\begin{minipage}{1\textwidth}
\begin{align}
    &\tilde{\mathcal{L}}_{\text{TOIB}}  \left(\{\boldsymbol{\phi}_i,\boldsymbol{\theta}_i, \boldsymbol{\psi}_{ij}\}^N_{i,j=1}\right) = \sum^N_{i=1} \frac{1}{V}\sum^{V}_{v=1}\frac{1}{L}\sum^L_{l=1}\left[ -\log q_{\boldsymbol{\theta}_i}\left(\boldsymbol{u}^{(v)}_i | \boldsymbol{y}^{(v,l)}_i\right) \right] + \beta \sum^N_{i=1} \frac{1}{V}\sum^{V}_{v=1} D_{\text{KL}}(p_{\boldsymbol{\phi}_i}(\boldsymbol{z}_i \mid \boldsymbol{x}^{(v)}_i) \parallel r_i(\boldsymbol{z}_i)) \notag \\
    & + \alpha \sum^{N}_{i=1}\sum^{N}_{\substack{j=1 \\ j \ne i}} \sum_{w\in\mathcal{W}_{\mathcal{B}}} \frac{|\mathcal{V}_w|}{V} \left( \underbrace{\overbrace{\frac{1}{|\mathcal{V}_w|} \sum_{v \in \mathcal{V}_w}  \log{{q}_{\boldsymbol{\psi}_{ij}}}\left( \boldsymbol{z}^{(v)}_{j} | \boldsymbol{z}^{(v)}_{i},w \right)}^{\text{matched term}} - \overbrace{\frac{1}{|\mathcal{V}_w|(|\mathcal{V}_w|-1)} \sum_{\substack{v' \ne v \\v,v' \in \mathcal{V}_w }} \log{{q}_{\boldsymbol{\psi}_{ij}}}\left( \boldsymbol{z}^{(v')}_{j} | \boldsymbol{z}^{(v)}_{i} ,w \right)}^{\text{mismatched term}}}_{\hat{I}_{\text{vCLUB}}(\boldsymbol{\psi}_{ij})}  \right). \label{eq: monte carlo loss}
\end{align}
\medskip
\hrule
\end{minipage}
\end{table*}

\subsection{The Optimisation Algorithm}
To obtain a tractable and differentiable approximation of the variational training objective in \eqref{eq: surrogate}, we employ Monte Carlo (MC) estimation of all expectations and optimise the parameters by stochastic gradient descent (SGD) with back‑propagation.
For each user $i$, we draw a mini-batch $\mathcal{B}=\left\{\left( 
\boldsymbol{x}^{(v)}_i, \boldsymbol{u}^{(v)}_i, \boldsymbol{w}^{(v)}_{ij} \right)\right\}^V_{v=1}$ with the mini-batch size $V$, sample encodings $\boldsymbol{z}^{(v)}_i \sim p_{\boldsymbol{\phi}_i}(\boldsymbol{z}_i \mid \boldsymbol{x}^{(v)}_i)$, and transmit the superimposed features through the downlink $L$ times to obtain the received signals $\boldsymbol{y}^{(v,l)}_i$. 
Since the latent factor $W_{ij}$ is assumed to be discrete in classification scenarios, we explicitly define its finite support set as $\mathcal{W} = \{1, \cdots, K\}$, where $K$ represents the total number of distinct task classes. Then, we introduce the batch-specific subset $\mathcal{W}_{\mathcal{B}} \subseteq \mathcal{W}$ containing only those categories present in the current mini-batch $\mathcal{B}$. For each category $w\in\mathcal{W}_{\mathcal{B}}$, we further define the sample set $\mathcal{V}_w = \{ v\in\{1,\cdots,V\}\mid \boldsymbol{w}^{(v)}_{ij} = w \}$, which represents the set of indices corresponding to samples within the mini-batch that belong to category $w$. Consequently, $|\mathcal{V}_w|$ denotes the number of samples associated with the same task category $w$ within in $\mathcal{B}$.
This construction matches the estimator structure shown in \eqref{eq: monte carlo loss}, where within-class matched and mismatched pairs are built by leave-one-out (LOO) inside each sample set $\mathcal{V}_w$. To handle cases where $|\mathcal{V}_w|<2$, we skip the mismatched term computation for numerical stability in practice.
Concretely, given a mini-batch of size $V$, denoted by $\mathcal{B}=\left\{\left( \boldsymbol{x}^{(v)}_i, \boldsymbol{u}^{(v)}_i \right)\right\}^V_{v=1}$ and assuming $L$ independent samples of the physical channel per data point, we define the Monte Carlo estimator of the loss function in \eqref{eq: monte carlo loss}, where $\boldsymbol{y}^{(v,l)}_i$ denotes the $l$-th sampled received signal corresponding to the $v$-th source information $\boldsymbol{x}^{(v)}_i$ for user $i$ after transmission through the channel. The $i$-th encoder output $\boldsymbol{z}^{(v)}_i \sim p_{\boldsymbol{\phi}_i}(\boldsymbol{z}_i \mid \boldsymbol{x}^{(v)}_i)$ is used both for downstream semantic decoding and for estimating cross-user semantic dependence via the CLUB objective $\hat{I}_{\text{CLUB}}(\boldsymbol{\psi}_{ij})$ in \eqref{eq: monte carlo loss}. 
Therefore, a variance-reduced MC estimator of the TOIB loss can be calculated by \eqref{eq: monte carlo loss}.

\subsection{Implementation Overhead and Computation Complexity}

For each mini-batch, the additional overhead of Algorithm~\ref{Algorithm 1} mainly stems from the pairwise orthogonality regularisation. In Phase-A, $N(N{-}1)$ lightweight CLUB estimators are updated for $M$ gradient-ascent steps, leading to $O(MN^2)$ auxiliary pair updates with respect to the number of users. In Phase-B, the same pairwise regulariser is further involved in the encoder--decoder back-propagation. Hence, the quadratic scaling over user pairs is the principal bottleneck for larger $N$. More precisely, under the estimator in~\eqref{eq: monte carlo loss}, the exact arithmetic cost also depends on the batch-level construction of matched and mismatched pairs. 



\begin{algorithm}[t]
\caption{The TOIB Training Procedure}
\textbf{Input:} Epochs $T$, batch size $V$, channel resamples $L$, CLUB steps $M$, negatives $K_{\text{neg}}$, weighting factors $\alpha,\beta$, initialised parameters $\{\boldsymbol{\phi}_i, \boldsymbol{\theta}_i, \boldsymbol{\psi}_{ij}\}_{i,j=1}^N$.
\begin{algorithmic}[1]
\For{$t=1$ \textbf{to} $T$}
    \For{$i=1$ \textbf{to} $N$}
      \State Sample mini-batch $\{(\boldsymbol{x}_i^{(v)},\boldsymbol{u}_i^{(v)})\}_{v=1}^V$;
      \State Generate encoder outputs $\{\boldsymbol{z}_i^{(v)}\}_{v=1}^{V}$ by \eqref{eq: encoder};
      \State Calculate superposed signal $\boldsymbol{s}^{(v)}$ by \eqref{eq: superimposed signal} and transmit over the AWGN channel;
      \State Resample channel $L$ times to generate received signals $\{\boldsymbol{y}_i^{(v,l)}\}_{v=1,l=1}^{V,L}$ by \eqref{eq: received signal};
      \State \textbf{Phase‑A} (update $\boldsymbol{\psi}_{ij}$): 
      \For{$m=1$ \textbf{to} $M$} 
        \For{$j=1$ \textbf{to} $N, j \ne i$} 
          \State Set $\mathcal V_w=\{v:\,w_{ij}^{(v)}=w\}$;
          \State Update CLUB-based estimator $\boldsymbol{\psi}_{ij}$ by gradient ascent on $\hat{I}_{\text{vCLUB}}(\boldsymbol{\psi}_{ij})$;
        \EndFor
      \EndFor
      \State \textbf{Phase‑B} (Back-propagation): Compute \eqref{eq: monte carlo loss} and update $\{\boldsymbol{\phi}_i,\boldsymbol{\theta}_i\}$ by back‑propagation method;
  \EndFor
\EndFor
\end{algorithmic}
\textbf{Output:} Optimised parameters $\{\boldsymbol{\phi}_i, \boldsymbol{\theta}_i, \boldsymbol{\psi}_{ij}\}_{i,j=1}^N$.
\label{Algorithm 1}
\end{algorithm}

\section{Simulation Results}
\label{sec: results}

\subsection{Experiments Setup}

We run the simulations on the CIFAR-10 dataset, which contains $50,000$ training data and $10,000$ testing data, to perform classification tasks. For fairness, we use ResNet-50 as the backbone network in all experiments. We set $\beta=0.01$ in all experiments with learning rate $10^{-4}$, and $\alpha=0.01$ in our proposed TOIB model. 
For fairness, our TOIB model and these three baselines are all trained with $\text{SNR}=5\;\text{dB}$ for $100$ epochs in a 2-user scenario, where U1 and U2 represent for user 1 and user 2, respectively. In addition, the equal power allocation strategy is performed.
We compare the proposed method with the following three schemes below: 
1) \textbf{Deep JSCC \cite{bourtsoulatze2019deep}}: An end-to-end DL-based JSCC scheme with the cross-entropy loss function. Since the original Deep JSCC is a single-user model, we superimpose the output of each users' encoder to generate the broadcast signal, which follows the same procedure in \eqref{eq: superimposed signal} and \eqref{eq: received signal}; 2) \textbf{Upper Bound}: We consider the original single-user model Deep JSCC as our upper bound to compare the difference between single-user model and multi-user model; 3) \textbf{Deep VIB}: If we set $\alpha=0$ in our proposed TOIB, we have a multi-user variational information bottleneck framework in \eqref{eq: conventional IB} with the cross-entropy and KL divergence loss function. Here, we assume the latent prior is a standard multivariable Gaussian distribution, i.e., $\boldsymbol{z}_i \sim \mathcal{N}(0, \boldsymbol{I})$.


\subsection{Experiments Results}
In Fig. \ref{fig: t-SNE}, we present the t-distribution stochastic neighbour embedding (t-SNE) of the each encoder's output \eqref{eq: encoder} in our proposed TOIB model and project it into a two-dimension latent space. The latent features of different users are expressed with different markers, and there are 10 classes in total. As shown in this figure, we can observe that the latent features of different classes are well-separated, which proves the classification ability of our TOIB model. We can also observe that both users' latent features have 10 classes, and latent features from different users are also well-separated, which proves the orthogonality ability of our TOIB model.

\begin{figure}[t]
\centerline{\includegraphics[width=0.45\textwidth]{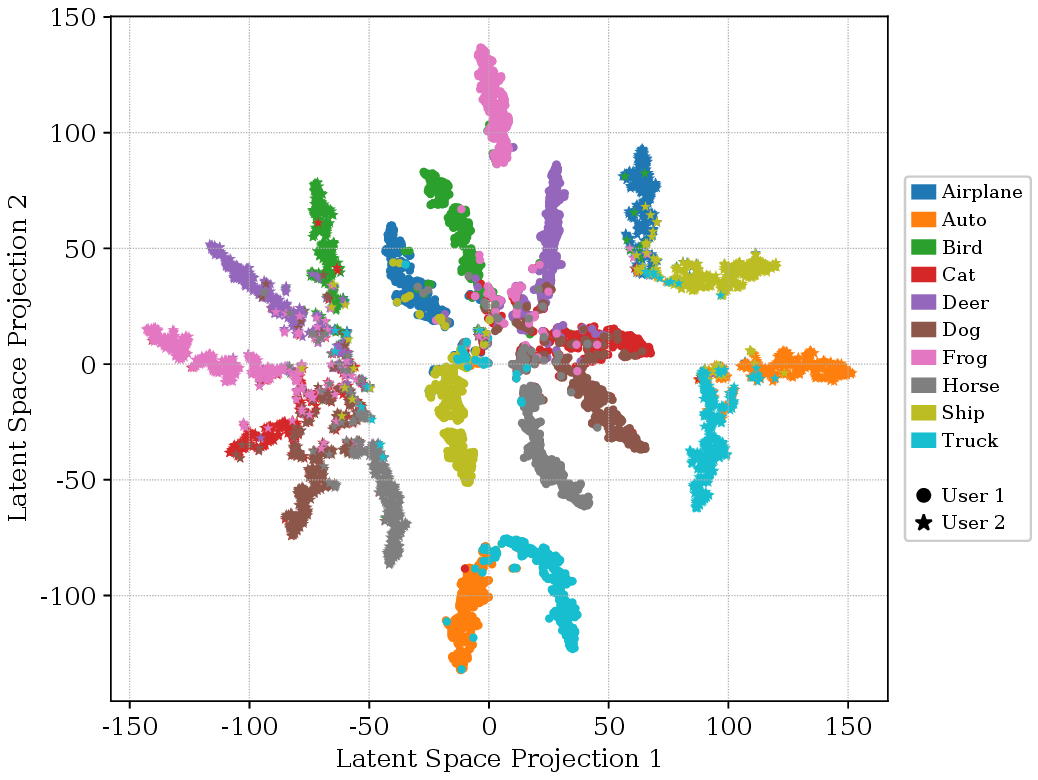}}
\caption{The t-SNE visualisation of each user's encoder output.}
    \label{fig: t-SNE}
\end{figure}

In Fig. \ref{fig: Acc AWGN}, we evaluate the classification accuracy over the AWGN channel. As shown in this figure, our proposed TOIB model performs better than other the Deep VIB model, where we set $\alpha=0$. This shows the effectiveness of our proposed orthogonality term in \eqref{eq: TOIB} compared with the traditional IB framework in \eqref{eq: conventional IB}. We also observe that our TOIB model performs better than the Deep JSCC model. This is due to the fact that the Deep JSCC model does not explicitly consider the orthogonality between semantic representations between different users, which may leads to semantic interference and reduce the classification accuracy. We can also observe that the curves corresponding to U1 and U2 across all models are closely aligned, which demonstrates the fairness of the proposed models.
In Fig. \ref{fig: Acc rayleigh}, we evaluate the classification accuracy over the Rayleigh fading channel. Our proposed TOIB model outperforms Deep VIB ($\alpha=0$) by $ 2\%-5\%$, demonstrating the effectiveness of the orthogonality term in \eqref{eq: TOIB}. TOIB also significantly outperforms Deep JSCC, which saturates at approximately $84\%$ due to the lack of explicit orthogonality constraints between users' semantic representations. The gap between TOIB and the single-user upper bound remains less than $1\%$, indicating that inter-user interference is effectively mitigated. 
The closely aligned U1 and U2 curves across all models further demonstrate the fairness of the proposed system.
\begin{figure}[t]
    \centerline{\includegraphics[width=0.35\textwidth]{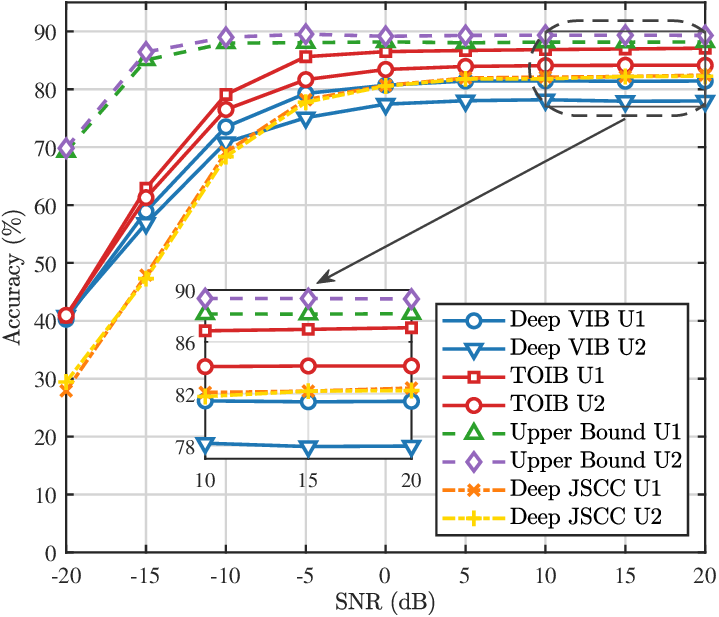}}
    \caption{The classification accuracy ($\%$) versus SNRs over AWGN channel.}
    \label{fig: Acc AWGN}
\end{figure}

\begin{figure}[t]
    \centerline{\includegraphics[width=0.35\textwidth]{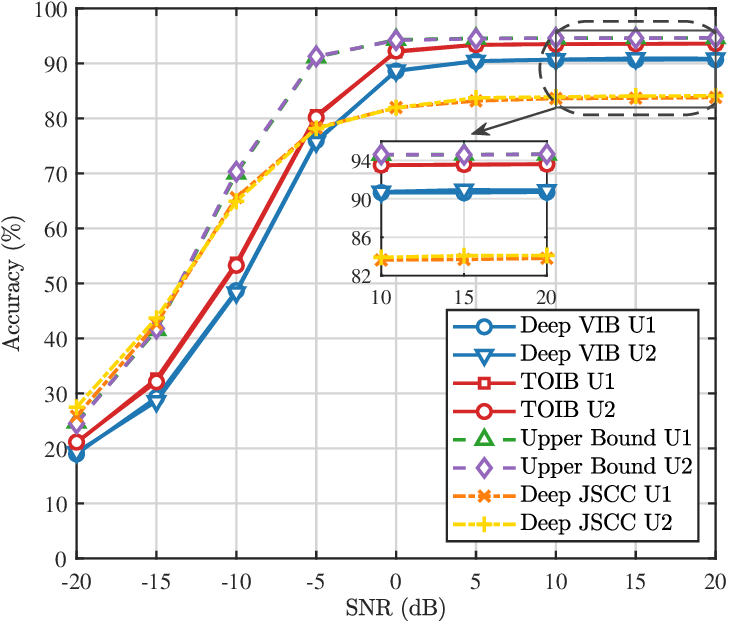}}
    \caption{The classification accuracy ($\%$) versus SNRs over Rayleigh channel.}
    \label{fig: Acc rayleigh}
\end{figure}

In Table \ref{tab: cross-decoding}, we evaluate the cross-decoding accuracy at $\text{SNR}=0~\text{dB}$ comparing across Deep JSCC, Deep VIB and our TOIB model. We use $\boldsymbol{D}_{\boldsymbol{\theta}_i}(\boldsymbol{y}_i)\!\rightarrow\! \boldsymbol{u}_i$ to represent for the normal decoding process, where the decoder $i$ tries to decode the information $\boldsymbol{u}_i$ from the received signal $\boldsymbol{y}_i$. On the contrary, we use $\boldsymbol{D}_{\boldsymbol{\theta}_i}(\boldsymbol{y}_i)\!\rightarrow\! \boldsymbol{u}_j, j \ne i$ to represent for the cross-decoding process, where the decoder $i$ tries to decode the information $\boldsymbol{u}_j$ from the received signal $\boldsymbol{y}_i$. As shown in the table, we observe that our proposed TOIB model performs better than the other two baselines in the normal decoding process, which is also demonstrated in Fig. \ref{fig: Acc AWGN}. We can also observe that the accuracy of cross decoding process is merely 10\%. Since we are evaluating a 10-class classification task, the accuracy level of cross-decoding is nearly equivalent to that of random guessing. 

\section{Conclusion}
\label{sec: conclusion}
This paper has proposed a novel TOIB framework for task-oriented distributed semantic communication systems. By introducing conditional task-related latent factors and incorporating a mutual information regulariser, TOIB explicitly enforces semantic orthogonality across users while maintaining semantic sufficiency and compression. A variational surrogate bound was derived to facilitate tractable optimization, employing CLUB-based estimators and Monte Carlo sampling.
Extensive experiments on classification tasks demonstrate that TOIB significantly improves accuracy over Deep JSCC and VIB baselines. 
These results underline TOIB as a promising direction for scalable and interference-resilient distributed semantic communication.

\vspace{-3.4mm}

\appendices
\section{Proof of Lemma \ref{theorem 1}} \label{appendix A}
We adopt the variational information bottleneck (VIB) method in \eqref{eq: loss TOIB} to approximate the intractable distributions $p(\boldsymbol{u}_i | \boldsymbol{y}_i)$, $p_{\boldsymbol{\phi}_i}(\boldsymbol{z}_i)$ and $p(\boldsymbol{z}_j|\boldsymbol{z}_i, \boldsymbol{w}_{ij})$ by $q_{\boldsymbol{\theta}_i}(\boldsymbol{u}_i | \boldsymbol{y}_i)$, $r_i(\boldsymbol{z}_i)$ and $q_{\boldsymbol{\psi}_{ij}}(\boldsymbol{z}_j|\boldsymbol{z}_i, \boldsymbol{w}_{ij})$, respectively.  
Specifically, the upper bound of $\mathcal{L}_{\text{TOIB}}$ in \eqref{eq: loss TOIB} is derived in \eqref{eq: vTOIB} at the top of this page.
where the inequality $(a)$ in \eqref{eq: vTOIB} holds since $D_{\text{KL}}( p_{\boldsymbol{\phi}_i}(\boldsymbol{u}_i | \boldsymbol{y}_i) \parallel q_{\boldsymbol{\theta}_i}(\boldsymbol{u}_i | \boldsymbol{y}_i) ) \geq 0$, $D_{\text{KL}}(p_{\boldsymbol{\phi}_i}(\boldsymbol{z}_i) \parallel r_i(\boldsymbol{z}_i)) \geq 0$ and the inequality \eqref{eq: third term CLUB} holds according to \cite{cheng2020club}:
\begin{align}
    I(Z_i, Z_j | \boldsymbol{w}_{ij}) \leq  I_{\text{CLUB}}(Z_i, Z_j | \boldsymbol{w}_{ij}). \label{eq: third term CLUB}
\end{align}
Since the distribution $p(\boldsymbol{z}_j |\boldsymbol{z}_i, \boldsymbol{w}_{ij})$ is intractable, we introduce a CLUB-based light-weight neural network with the parameter $\boldsymbol{\psi}_{ij}$ to approximate the surrogate upper bound in \eqref{eq: vTOIB}, i.e., $q_{\boldsymbol{\psi}_{ij}}(\boldsymbol{z}_j |\boldsymbol{z}_i, \boldsymbol{w}_{ij}) \approx p(\boldsymbol{z}_j |\boldsymbol{z}_i, \boldsymbol{w}_{ij})$. We claim that the inequality in \eqref{eq: surrogate} means that $I_{\text{vCLUB}}(Z_i,Z_j|\boldsymbol{w}_{ij})$ no longer guarantees an upper bound of $I(Z_i,Z_j|\boldsymbol{w}_{ij})$, but with a good variational approximation $q_{\boldsymbol{\psi}_{ij}}(\boldsymbol{z}_j |\boldsymbol{z}_i, \boldsymbol{w}_{ij})$, $I_{\text{vCLUB}}(Z_i,Z_j|\boldsymbol{w}_{ij})$ can still hold a MI upper bound or become a reliable MI estimator, which is illustrated detailed in \cite{cheng2020club}. Therefore, as long as $q_{\boldsymbol{\psi}_{ij}}(\boldsymbol{z}_j |\boldsymbol{z}_i, \boldsymbol{w}_{ij})$ can adequately approximate $p(\boldsymbol{z}_j |\boldsymbol{z}_i, \boldsymbol{w}_{ij})$, the inequality \eqref{eq: surrogate} behaves as a reliable estimator practically.

\begin{table*}[!t]
\centering
\begin{minipage}{1\textwidth}
\begin{subequations}
\begin{align}
    \mathcal{L}_{\text{TOIB}} 
    & \overset{(a)}{\leq} \sum^N_{i=1} \left[ \mathbb{E}_{p_{\boldsymbol{\phi}_i}(\boldsymbol{u}_i, \boldsymbol{y}_i)}[-\log q_{\boldsymbol{\theta}_i}(\boldsymbol{u}_i \mid \boldsymbol{y}_i)] - \underbrace{H(U_i)}_{constant} + \beta \mathbb{E}_{p_{\boldsymbol{\phi}_i}(\boldsymbol{x}_i)}\left[D_{\text{KL}}(p_{\boldsymbol{\phi}_i}(\boldsymbol{z}_i \mid \boldsymbol{x}_i) \parallel r_i(\boldsymbol{z}_i))\right] \right] \notag \\
    & + \alpha \sum^N_{i=1} \sum^N_{\substack{j=1 \\ j \ne i}} \mathbb{E}_{p(\boldsymbol{w}_{ij})} \left[\underbrace{\mathbb{E}_{p(\boldsymbol{z}_i, \boldsymbol{z}_j | \boldsymbol{w}_{ij})} [{\log p(\boldsymbol{z}_j | \boldsymbol{z}_i, \boldsymbol{w}_{ij})}]-\mathbb{E}_{p(\boldsymbol{z}_i|\boldsymbol{w}_{ij})p(\boldsymbol{z}_j|\boldsymbol{w}_{ij})}[{\log  p(\boldsymbol{z}_j | \boldsymbol{z}_i, \boldsymbol{w}_{ij})}]}_{I_{\text{CLUB}(Z_i,Z_j|\boldsymbol{w}_{ij})}} \right] \label{eq: vTOIB} \\
    & \overset{(b)}{\simeq}  \sum^N_{i=1} \left[ \mathbb{E}_{p_{\boldsymbol{\phi}_i}(\boldsymbol{u}_i, \boldsymbol{y}_i)}[-\log q_{\boldsymbol{\theta}_i}(\boldsymbol{u}_i \mid \boldsymbol{y}_i)] + \beta \mathbb{E}_{p_{\boldsymbol{\phi}_i}(\boldsymbol{x}_i)}\left[D_{\text{KL}}(p_{\boldsymbol{\phi}_i}(\boldsymbol{z}_i \mid \boldsymbol{x}_i) \parallel r_i(\boldsymbol{z}_i))\right] \right] \notag \\
    & + \alpha \sum^N_{i=1} \sum^N_{\substack{j=1 \\ j \ne i}} \mathbb{E}_{p(\boldsymbol{w}_{ij})} \left[\underbrace{\mathbb{E}_{p(\boldsymbol{z}_i, \boldsymbol{z}_j | \boldsymbol{w}_{ij})} [{\log q_{\boldsymbol{\psi}_{ij}}(\boldsymbol{z}_j|\boldsymbol{z}_i, \boldsymbol{w}_{ij})}]-\mathbb{E}_{p(\boldsymbol{z}_i|\boldsymbol{w}_{ij})p(\boldsymbol{z}_j|\boldsymbol{w}_{ij})}[{\log  q_{\boldsymbol{\psi}_{ij}}(\boldsymbol{z}_j|\boldsymbol{z}_i, \boldsymbol{w}_{ij})}]}_{I_{\text{vCLUB}(Z_i,Z_j|\boldsymbol{w}_{ij})}} \right] = \mathcal{L}_{\text{vTOIB}} \label{eq: surrogate}
\end{align}
\end{subequations}
\medskip
\hrule
\end{minipage}
\end{table*}

\begin{table}[t]
    \centering
    \caption{Cross-decoding accuracy (\%) at different SNRs}
    \label{tab: cross-decoding}
    \renewcommand{\arraystretch}{1.2}
    \resizebox{\columnwidth}{!}{%
    \begin{tabular}{|c|c|c|c|c|c|}
        \hline
        
        & $\boldsymbol{D}_{\boldsymbol{\theta}_1}(\boldsymbol{y}_1)\!\rightarrow\! \boldsymbol{u}_1$
        & $\boldsymbol{D}_{\boldsymbol{\theta}_1}(\boldsymbol{y}_1)\!\rightarrow\! \boldsymbol{u}_2$
        & $\boldsymbol{D}_{\boldsymbol{\theta}_2}(\boldsymbol{y}_2)\!\rightarrow\! \boldsymbol{u}_2$
        & $\boldsymbol{D}_{\boldsymbol{\theta}_2}(\boldsymbol{y}_2)\!\rightarrow\! \boldsymbol{u}_1$ 
        \\
        \hline
        Deep JSCC    & 82.03\% &  9.96\% & 81.83\% & 10.88\%  \\
        \hline
        Deep VIB     & 81.59\% &  9.98\% & 77.98\% & 10.08\%  \\
        \hline
        TOIB         & \textbf{87.04\%} & 10.18\% & \textbf{84.05\%} & 10.40\%  \\
        \hline
    \end{tabular}%
    }
\end{table}

\bibliographystyle{ieeetr}
\bibliography{ref}
\end{document}